\documentclass[letterpaper,11pt]{article}

\usepackage{fullpage}
\usepackage{graphicx}
\usepackage{subfigure}
\usepackage{amsmath,amssymb,amsfonts,amsthm}
\usepackage{xspace}
\usepackage{hyperref}

\newif\ifCtableInstalled
\CtableInstalledfalse
\ifCtableInstalled
  \usepackage{multirow,ctable}
\else
\fi

\newcommand{\Vol}{\mathop{\rm vol}}
\newcommand{\poly}{\mathop{\rm poly}}
\newcommand{\polylog}{\mathop{\rm polylog}}
\DeclareMathOperator{\ar}{asp}
\newcommand{\RAR}{\ar_{\rm rect}}

\newcommand{\tree}{\ensuremath{\mathcal{T}}}
\newcommand{\node}{\nu}
\newcommand{\othernode}{\mu}
\DeclareMathOperator{\myroot}{root}

\DeclareMathOperator{\area}{area}

\DeclareMathOperator{\bd}{\partial}
\DeclareMathOperator{\diam}{diam}
\DeclareMathOperator{\myangle}{angle}
\newcommand{\dist}[2]{|#1 #2|}

\DeclareMathOperator{\mydepth}{height}
\newcommand{\weight}{w}

\newcommand{\eps}{\varepsilon}
\newcommand{\gb}{8}            % exponent in the bound for the greedy method
\newcommand{\opt}{\mbox{{\sc opt}}\xspace}

\newcommand{\Reals}{{\mathbb{R}}}
\newcommand{\etal}{\emph{et~al.}\xspace}

\newtheorem{theorem}{Theorem}
\newtheorem{lemma}{Lemma}

\newtheorem{observation}{Observation}

\title{Fat Polygonal Partitions with Applications\\to Visualization
and Embeddings\thanks{A preliminary version of
this paper appeared in the Proceedings of the 24th ACM Symposium
on Computational Geometry~\cite{conference_version}.}}

\author{%
Mark de Berg\thanks{TU Eindhoven. Email: {\tt mdberg@win.tue.nl}.}
\and
Krzysztof Onak\thanks{IBM T.J.\ Watson Research Center. Email: {\tt krzysztof@onak.pl}. Supported in part by NSF grants 0514771, 0728645, and 0732334. The research was done when the author was at a student at MIT.}
\and
Anastasios Sidiropoulos\thanks{Dept.~of Computer Science \& Engineering, and Dept.~of Mathematics, The Ohio State University. Email: {\tt sidiropoulos.1@osu.edu}.}
}

\date{December 13, 2013}

\begin{document}

\maketitle

\begin{abstract}
Let $\tree$ be a rooted and weighted tree, where the weight of any node is
equal to the sum of the weights of its children.
The popular Treemap algorithm visualizes such a tree as
a hierarchical partition of a square into rectangles, where the area of the rectangle corresponding to any
node in $\tree$ is equal to the weight of that node. The aspect ratio of the rectangles in such
a rectangular partition necessarily depends on the weights and can become arbitrarily high.

We introduce a new hierarchical partition scheme, called a \emph{polygonal partition}, which uses
convex polygons rather than just rectangles. We present two methods for constructing polygonal partitions,
both having guarantees on the worst-case aspect ratio of the constructed polygons; in particular, both
methods guarantee a bound on the aspect ratio that is independent of the weights of the nodes.

We also consider \emph{rectangular partitions with slack}, where
the areas of the rectangles may differ slightly from the
weights of the corresponding nodes. We show that this makes it possible to obtain partitions
with constant aspect ratio. This result generalizes to hyper-rectangular partitions
in $\Reals^d$. We use these partitions with slack for
embedding ultrametrics into $d$-dimensional Euclidean space:  we give a $\polylog(\Delta)$-approximation
algorithm for embedding $n$-point ultrametrics into $\mathbb R^d$ with minimum distortion, where $\Delta$ denotes
the spread of the metric, i.e., the ratio between the largest and the smallest distance between two points. The previously best-known approximation ratio for this problem was polynomial in~$n$.
This is the first algorithm for embedding a non-trivial family of weighted-graph metrics into a space of
constant dimension that achieves polylogarithmic approximation ratio.
\end{abstract}

\section{Introduction}
Hierarchical structures are commonplace in many areas. It is not surprising, therefore,
that the visualization of hierarchical structures---in other words, of rooted trees---is
one of the most widely studied problems in information visualization and graph drawing.
In the weighted variant of the problem, we are given a rooted tree in which
each leaf has a positive weight and the weight of each internal node is the sum of the
weights of the leaves in its subtree. One of the most successful practical algorithms for
visualizing such weighted trees is the so-called Treemap algorithm.
Treemap visualizes the given tree by constructing a hierarchical \emph{rectangular partition}
of a square, as illustrated in Figure~\ref{fig:sample_rect}. More precisely, Treemap
assigns a rectangle to each node in the tree such that
\begin{itemize}
\item
the area of the rectangle is equal to the weight of the node;
\item
the rectangles of the children of each internal node $\node$ form
a partition of the rectangle of~$\node$.
\end{itemize}
The Treemap algorithm was proposed by Shneiderman~\cite{treemap_main} and its first efficient implementation
was given by Johnson and Shneiderman~\cite{treemap_2}.
Treemap has been used to visualize a wide range of hierarchical data,
including stock portfolios~\cite{treemap_stock}, news items~\cite{newsmap}, blogs~\cite{treemap_blogs},
business data~\cite{treemap_business}, tennis matches~\cite{treemap_tennis}, photo collections~\cite{treemap_photo},
and file-system usage~\cite{treemap_main, circular}.
Shneiderman maintains a webpage \cite{treemap_history} that describes the history of his invention and
gives an overview of applications and proposed extensions to his original idea. Below we only discuss the
results that are directly related to our work.

In general, there are many different rectangular partitions corresponding to a given tree.
To obtain an effective visualization it is desirable that the aspect ratio of the rectangles
be kept as small as possible; this way the individual rectangles are easier to distinguish and
the areas of the rectangles are easier to estimate. Various heuristics have been proposed for minimizing the aspect ratio
of the rectangles in the partition~\cite{treemap_sq,treemap_ordered,turo_johnson92}.
Unfortunately, the aspect ratio can become arbitrarily bad if the weights have unfavorable values.
For example, consider a tree with a root and two leaves, where the first leaf has weight~1 and the second has weight $W$.
Then the optimal aspect ratio of any rectangular partition is unbounded as $W\rightarrow \infty$.
Hence, in order to obtain guarantees on the aspect ratio we cannot restrict ourselves to rectangles.
This lead Balzer~et~al.~\cite{voronoi1,voronoi2} to introduce Voronoi treemaps, which use
more general regions in the partition. However, their approach is heuristic and it does
not come with any guarantees on the aspect ratio of the produced regions.
Thus the following natural question is still open:
\begin{quote}
Suppose we are allowed to use arbitrary convex polygons in the partition, rather than just rectangles.
Is it then always possible to obtain a partition that achieves aspect ratio independent of the weights
of the nodes in the input tree? (The aspect ratio of a convex region~$A$ is defined as $\diam(A)^2/\area(A)$, where $\diam(A)$ is its diameter and $\area(A)$ is its area.\footnote{Another common definition of the aspect ratio of a convex region $A$ is the ratio between the radius $R$ of the smallest circumscribing circle and the radius $r$ of the largest inscribed circle. The aspect ratio defined in this manner is sometimes referred to as the \emph{fatness} of the region. There are several other definitions of fatness, all of which are equivalent up to constant factors for convex planar objects~\cite{bksv-rimga-02}.
In particular, in our case it is easy to show that $\diam(A) = \Theta(R)$ and $\area(A) = \Theta(R \cdot r)$, which implies that $\diam(A)^2/\area(A) = \Theta(R/r)$.
})
\end{quote}

\begin{figure}[t]
\begin{center}
\subfigure[A rectangular partition\label{fig:sample_rect}]{
\scalebox{0.24}{\includegraphics{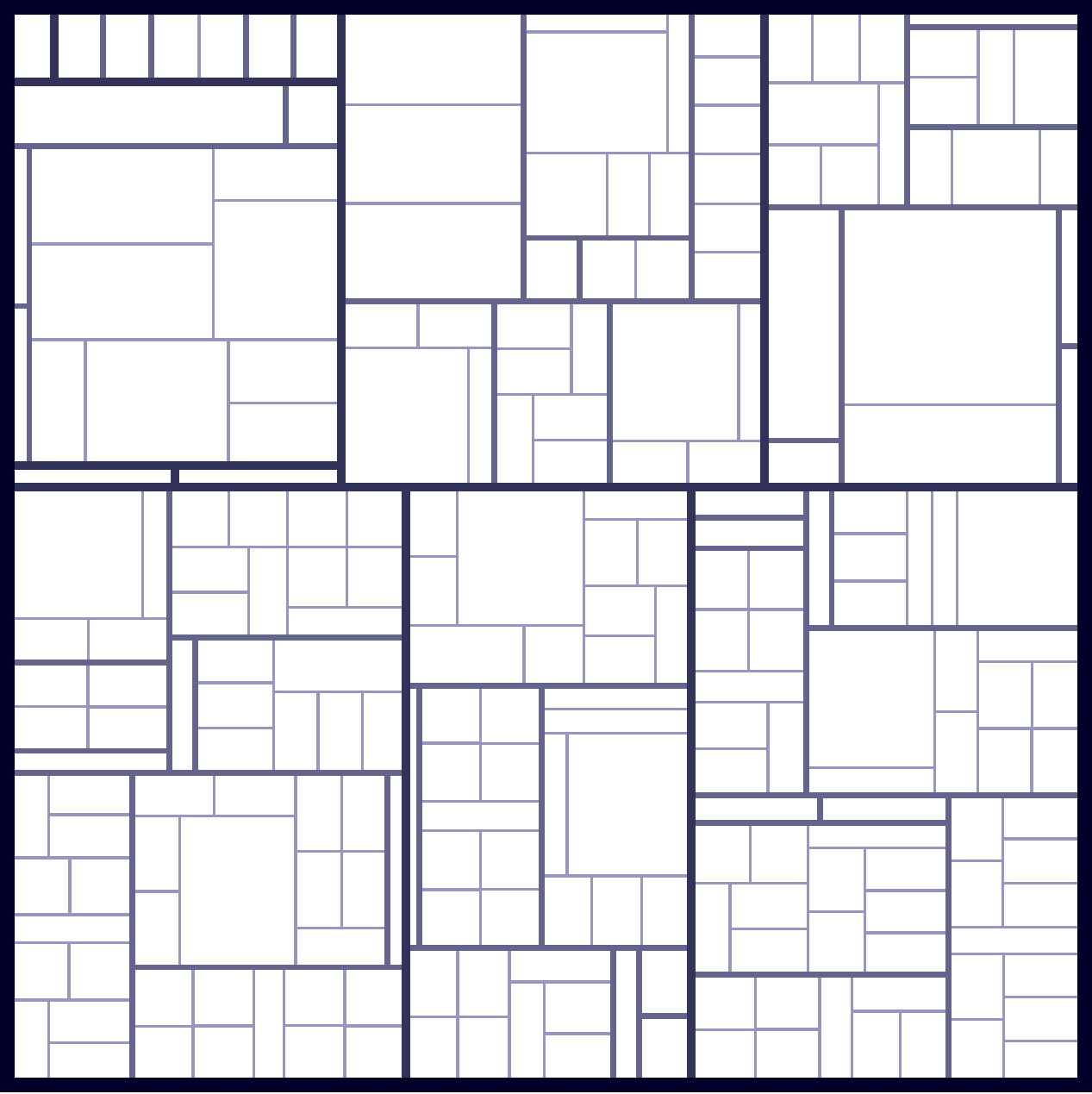}}
}
\subfigure[A greedy polygonal partition\label{fig:sample_greedy}]{
\scalebox{0.24}{\includegraphics{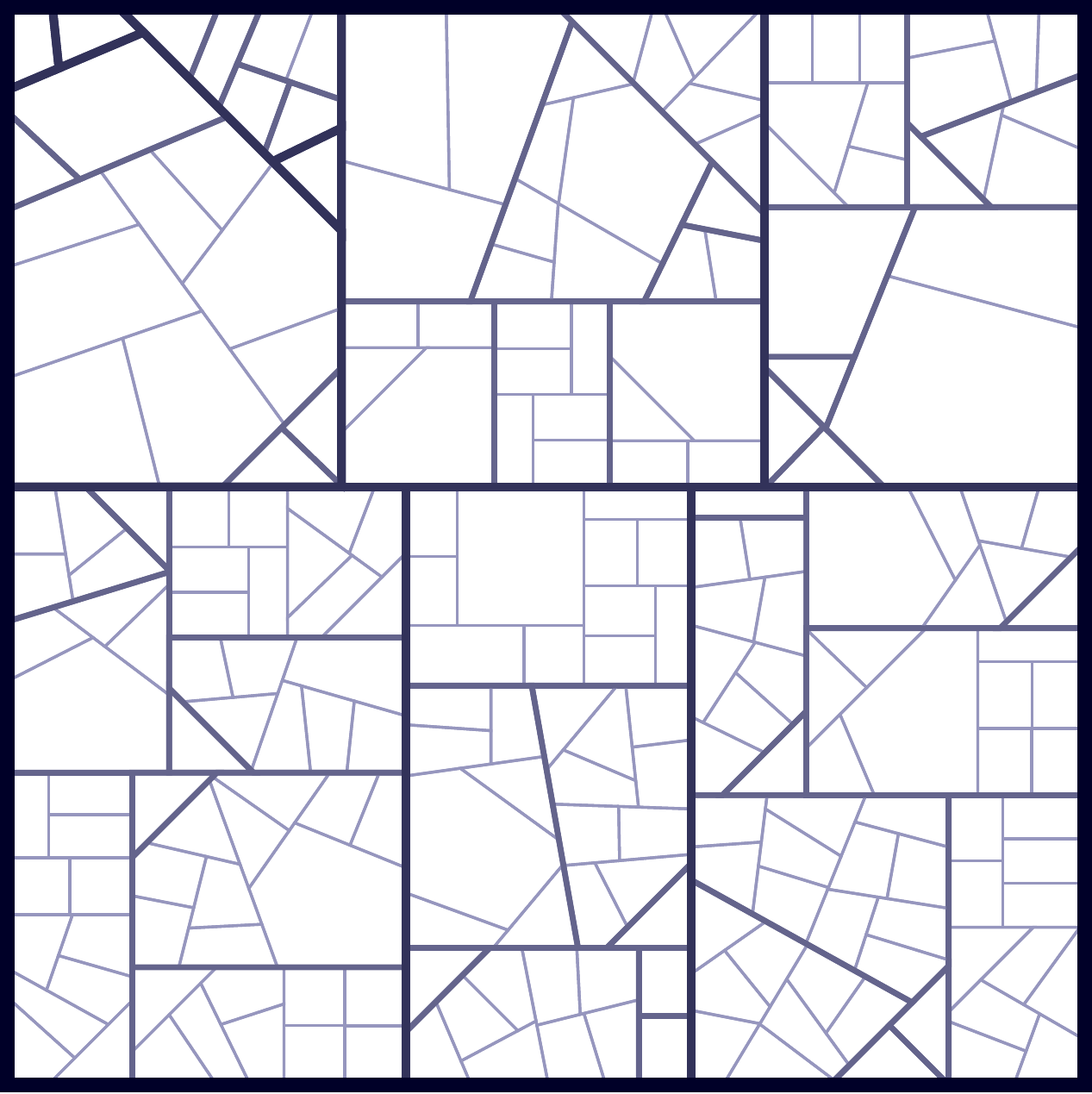}}
}
\subfigure[An angular polygonal partition\label{fig:sample_angular}]{
\scalebox{0.24}{\includegraphics{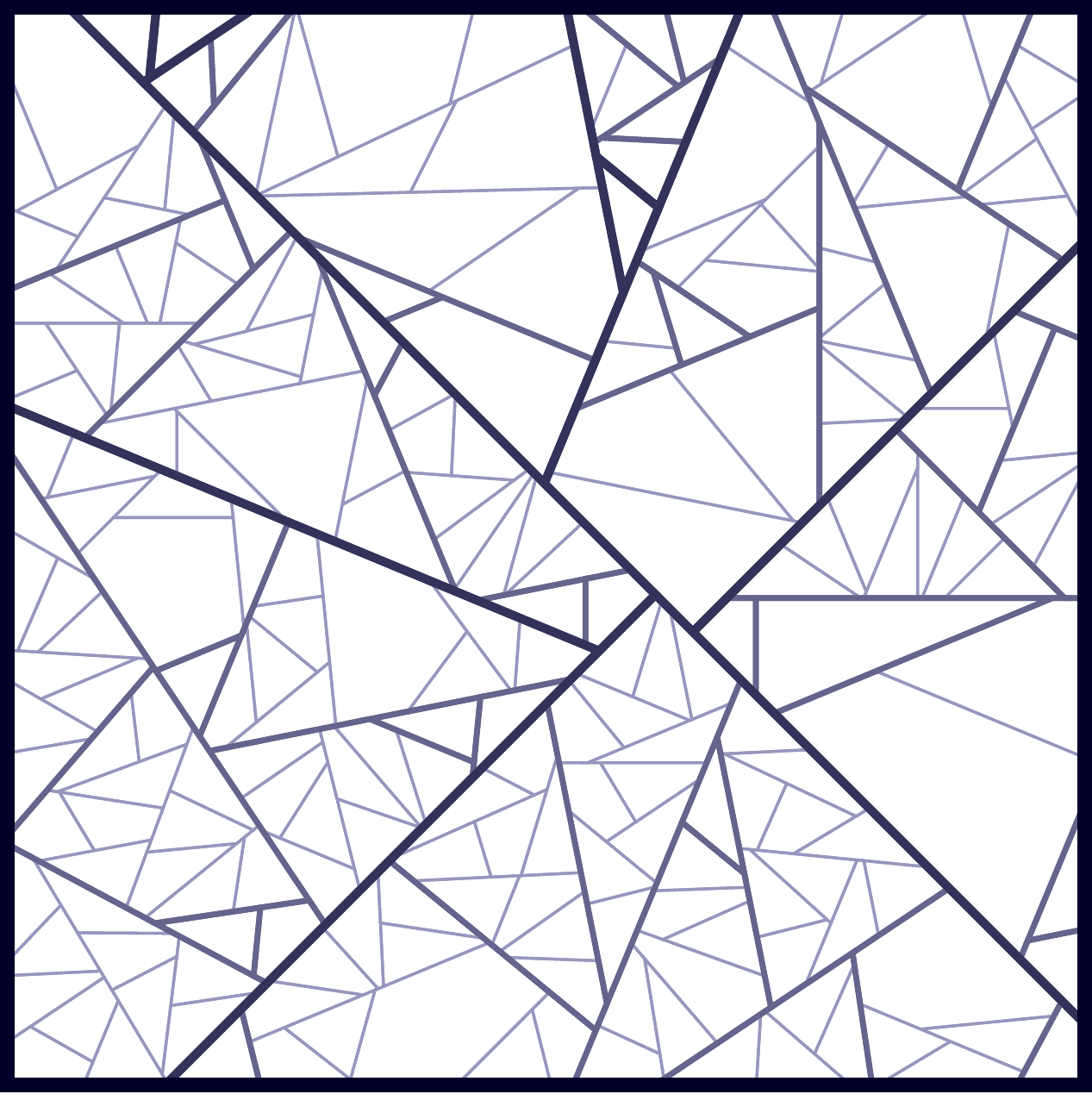}}
}
\caption{Sample polygonal partitions}
\end{center}
\end{figure}

\paragraph{Our results on hierarchical partitions.}
Our main result is an affirmative answer to the question above:
we present two algorithms that, given an $n$-node tree of height $h$, construct a partition
into convex polygons with aspect ratio $O(\poly(h,\log n))$.
Our algorithms, which are described in Section~\ref{sect:two_methods},
are very simple. They first convert the input tree into a binary tree,
and then recursively partition the initial square region using straight-line cuts.
The methods differ in the way in which the orientation of the cutting line is chosen
at each step. The \emph{greedy method} minimizes the maximum aspect ratio of two subpolygons
resulting from the cut, while the \emph{angular method} maximizes the angle that the splitting line makes with
any of the edges of the polygon being cut. Figures~\ref{fig:sample_greedy}~and~\ref{fig:sample_angular}
depict partitions computed by our algorithms.

The main challenge lies in the analysis of the aspect ratio achieved by the algorithms,
which is given in Sections~\ref{sect:angular_analysis}~and~\ref{sect:greedy_analysis}.
We prove that the angular method produces a partition with aspect ratio
$O(h+\log n)$.
For the greedy method we can only prove an aspect ratio of $O((h+\log n)^{\gb})$. Since the
greedy method is the most natural one, we believe this result is still interesting.
Moreover, in the (limited) experiments we have done---see Section~\ref{sect:two_methods}---the
greedy method always outperforms the angular method.
Besides these two algorithms, we also prove a lower bound: we show
in Section~\ref{sect:lower-bound} that for certain trees
and weights, any partition into convex polygons must have polygons with aspect ratio~$\Omega(h)$.\footnote{Recently de Berg, Speckmann, and van der Weele~\cite{BergSW2010} refined our
angular method to get rid of the additive $O(\log n)$ factor in the upper bound, thus obtaining a polygonal partition of aspect ratio $O(h)$.}

After having studied the problem of constructing polygonal partitions, we return to
rectangular partitions. As observed, it is in general not possible to obtain any guarantees on
the aspect ratio of the rectangles in the partition.
If, however, we are willing to let the area of a rectangle deviate slightly from the
weight of its corresponding node then we can obtain bounded aspect ratio,
as we show in Section~\ref{sect:slack}.
More precisely, we obtain the following partition.
Let $\eps \in (0,1/3)$.
We allow that the area $A$ of the rectangle assigned to every non-root node $v$ is shrunken by a factor of at most $1-\eps$ compared to its share of the area $A'$ of the rectangle of the parent $v'$. That is, we only require that
$$(1-\eps)\cdot \frac{A'}{w_{v'}} \le \frac{A}{w_v} \le \frac{A'}{w_{v'}},$$
where $w_v$ and $w_{v'}$ are the weights of $v$ and $v'$, respectively.
Then we show that the aspect ratio of every rectangle can be bounded by $1/\eps$.
We call this kind of partition a \emph{rectangular partition with slack}.

\paragraph{Application to embedding ultrametrics.}
The work of B\u{a}doiu~\etal~\cite{BadoiuSOCG2006} establishes a lower bound for the distortion of the best embedding of an ultrametric into $\mathbb R^d$. Our hierarchical partitions with low aspect ratio can be used for efficiently constructing an embedding that closely matches the lower bound of B\u{a}doiu~\etal More details, including a brief history of relevant embedding results, follow.

Let us first recall a few standard definitions.
A metric space $M=(X,D)$ is a set $X$ together with a symmetric distance function $D\colon X\times X \rightarrow \Reals_{\geq 0}$ that satisfies the triangle inequality and
$D(x_1,x_2) = 0$ if and only if $x_1 = x_2$.
An \emph{embedding} of a metric space $M=(X,D)$ into a \emph{host} metric space $M'=(X',D')$ is an injective mapping $f\colon X\to X'$.
The \emph{distortion} of an embedding $f$ is defined as
\[
\max_{x,y\in X} \frac{D'(f(x),f(y))}{D(x,y)} \cdot \max_{x,y\in X}\frac{D(x,y)}{D'(f(x),f(y))}.
\]
Over the past few decades, low-distortion embeddings of metric spaces into various host spaces have been the subject of extensive study \cite{I-survey}.
Embeddings into Euclidean space are of particular importance in applications, and have received a lot of attention.
Bourgain's theorem~\cite{Bou} asserts that any $n$-point metric space admits an embedding into high-dimensional Euclidean space with distortion $O(\log n)$.
Matou\v{s}ek~\cite{Matousek-smalld} has shown that the minimum distortion for embedding into $d$-dimensional Euclidean space is $n^{\Theta(1/d)} \cdot \log^{O(1)} n$.
Since the distortion for embedding into constant-dimensional Euclidean space can be polynomially large in the worst case, it is natural to ask whether we can \emph{approximate} the best possible distortion for a given input metric.
Matou\v{s}ek and Sidiropoulos \cite{sidiropo_hardness} have shown that minimum-distortion embeddings of general metrics into $\Reals^d$ (with $d\geq 2$) are hard to approximate to within a factor of roughly $n^{1/(22d-10)}$, unless {\sc p=np}~\cite{sidiropo_hardness}.
In other words, it is unlikely that there exists a polynomial-time algorithm with significantly better performance than the worst case guarantee.

In light of the above inapproximability result, it is natural to ask whether there exist interesting families of metrics, for which we can obtain better than polynomial approximation factors for embedding into constant-dimensional Euclidean space.
In this paper we present the first result of this type, for embedding \emph{ultrametrics} into $\Reals^d$.
An ultrametric is a metric satisfying the following strengthened version of
the triangle inequality: for any $x,y,z\in X$ we have $D(x,z) \leq \max\{D(x,y),D(y,z)\}$.
Equivalently,  $M=(X,D)$ is an ultrametric if it can be realized as
the shortest-path metric over the leaves of a rooted edge-weighted tree
such that the distance between the root and any leaf is the same.
Ultrametrics have received a lot of attention in the embeddings literature, and play a central role in many algorithmic applications (see e.g.~\cite{Bar96}).

B\u{a}doiu~\etal~\cite{BadoiuSOCG2006} showed that finding a minimum-distortion embedding of an ultrametric into $\mathbb{R}^2$ is {\sc np}-complete and presented
an $O(n^{1/3})$-approximation algorithm for the problem.
They extended the algorithm to embedding ultrametrics into $\mathbb{R}^d$, obtaining an $(n^{\frac{1}{d}-\Theta(\frac{1}{d^2})})$-approximation.
This result is obtained using a lower bound on the amount of space required
in a non-contracting embedding of every subtree of an ultrametric into $\mathbb R^d$.
We apply our results on rectangular partitions with slack
to a hierarchical structure corresponding to the ultrametric with weights given by the lower bound.
Bounded aspect ratios in our partition imply relatively low distortion and good approximation to the best embedding of the ultrametric
into $\mathbb R^d$.
Moreover, we show a connection between embedding ultrametrics into $\Reals^d$ and
\mbox{(hyper-)}rectangular partitions with slack and we use this connection
to obtain a significant improvement over the result of B\u{a}doiu~\etal~\cite{BadoiuSOCG2006}.
More precisely, using our results on (hyper-)rectangular partitions with slack,
we obtain a polynomial-time $\polylog(\Delta)$-approximation algorithm for the problem of
embedding ultrametrics into $(\mathbb{R}^d, \ell_2)$ with minimum distortion, where $\Delta$
is the \emph{spread} of $X$. (The spread of $X$ is defined as
$\Delta = \diam(X) / \min_{x,y\in X}D(x,y)$.)
As long as the spread is sub-exponential in $n$, this is an exponential improvement over \cite{BadoiuSOCG2006}.

\section{The two algorithms}\label{sect:two_methods}
Before we present our algorithms, we define the problem more formally and introduce
some notation.
Let $\tree$ be a rooted tree. We say that $\tree$ is \emph{properly weighted} if
each node $\node\in\tree$ has a positive weight $\weight(\node)$ that
equals the sum of the weights of the children of~$\node$.
% In other words, $\weight(\node)$ equals the sum of the weights of the leaves in the subtree rooted at~$\node$.
We assume without loss of generality that $\weight(\myroot(\tree))=1$.
A \emph{polygonal partition} for a properly weighted tree assigns
a convex polygon $P(\node)$ to each node $\node\in\tree$ such that
\begin{itemize}
\item the polygon $P(\myroot(\tree))$ is the unit square;
\item for any node $\node$ we have $\area(P(\node)) = \weight(\node)$;
\item for any node $\node$, the polygons assigned to the children of $\node$ form a disjoint partition of $P(\node)$.
\end{itemize}
Recall that the aspect ratio of a planar convex region $A$, denoted
by $\ar(A)$,  is defined as $\ar(A) := \diam(A)^2/\area(A)$.
The aspect ratio of a polygonal partition is the maximum aspect ratio of any of the
polygons in the partition. Our goal is to show that any properly weighted tree
admits a \emph{fat polygon partition}, that is, a polygonal partition with small aspect ratio.
\medskip

We propose two methods for constructing fat polygonal partitions.
They both start with transforming the input tree~$\tree$ into a \emph{binary} tree~$\tree'$.
The nodes of $\tree$ are a subset of nodes of $\tree'$, and
two nodes are in the ancestor-descendant relation in $\tree$
if and only if they are in the same relation in $\tree'$.
The weights assigned to nodes of $\tree$ are preserved in $\tree'$.
Any polygonal partition for $\tree'$ restricted to nodes from $\tree$
is a polygonal partition for $\tree$.
Then, for the binary tree $\tree'$, it suffices to design a method
that cuts the polygon $P(\node)$ corresponding to a node $\node$ into two polygons
of prespecified areas that correspond to $\node$'s children. We propose two such methods:
the \emph{angular} method and the \emph{greedy} method. To achieve a polygonal
partition for $\tree'$, it suffices to recursively apply one of the
cutting methods.

\paragraph{The transformation into a binary tree.}
We transform the input $n$-node tree $\tree$ into a binary tree $\tree'$
by replacing every internal node $\node$ of degree greater than two
by a collection of nodes whose subtrees together are exactly
the subtrees of~$\node$. This can be done
in such a way that $\mydepth(\tree')=O(\mydepth(\tree)+\log n)$~\cite{ps-cgi-85}.
The number of nodes in $\tree'$ is $O(n)$.
For completeness we sketch how this transformation is done.

For a node $\node$, we use $\tree_\node$ to denote the subtree rooted at $\node$,
and $|\tree_\node|$ to denote the number of nodes in~$\tree_\node$.
The transformation is a recursive process, starting at the root of~$\tree$.
Suppose we reach a node~$\node$. If $\node$ has degree two or less, we just recurse
on the at most two children of~$\node$. If $\node$ has degree $k\geq 3$, we proceed
as follows. Let $C(\node)$ be the set of children of~$\node$, and let $\othernode\in C(\node)$
be the child with the largest number of nodes in its subtree. We partition
$C(\node)\setminus \{\othernode\}$ into two non-empty subsets $C_1(\node)$ and $C_2(\node)$ such that
$\sum_{\othernode\in C_1(\node)} |\tree_\othernode| < |\tree_\node|/2$ and
$\sum_{\othernode\in C_2(\node)} |\tree_\othernode| < |\tree_\node|/2$.
We create three new nodes $\node_1,\node_2,\node_3$ and modify the tree
as shown in Figure~\ref{fi:transform-to-binary}. The weights
$\weight(\node_1),\weight(\node_2),\weight(\node_3)$ are set to the sum of the
weights of the leaves in their respective subtrees.
Finally, we recurse on $\node_1$, $\othernode$, and~$\node_3$.
\begin{figure}[t]
\begin{center}
\includegraphics{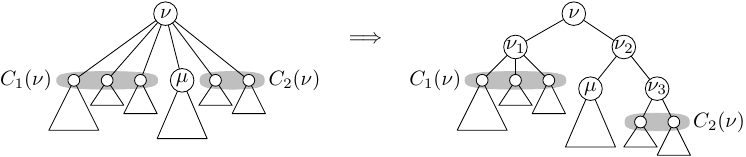}
\caption{Transforming a tree to a binary tree.}
\label{fi:transform-to-binary}
\end{center}
\end{figure}
After the procedure has finished we have a (properly weighted)
tree in which every node has degree at most two.
We remove all degree-1 nodes to obtain our binary tree~$\tree'$.
The height of $\tree'$ is at most $2(\mydepth(\tree)+\log n)$,
because every time we go down two levels in $\tree'$ we either pass
through an original node from $\tree$ or the number of nodes in the
subtree halves.

\begin{figure}[t]
\begin{center}
\includegraphics[width=3.5in]{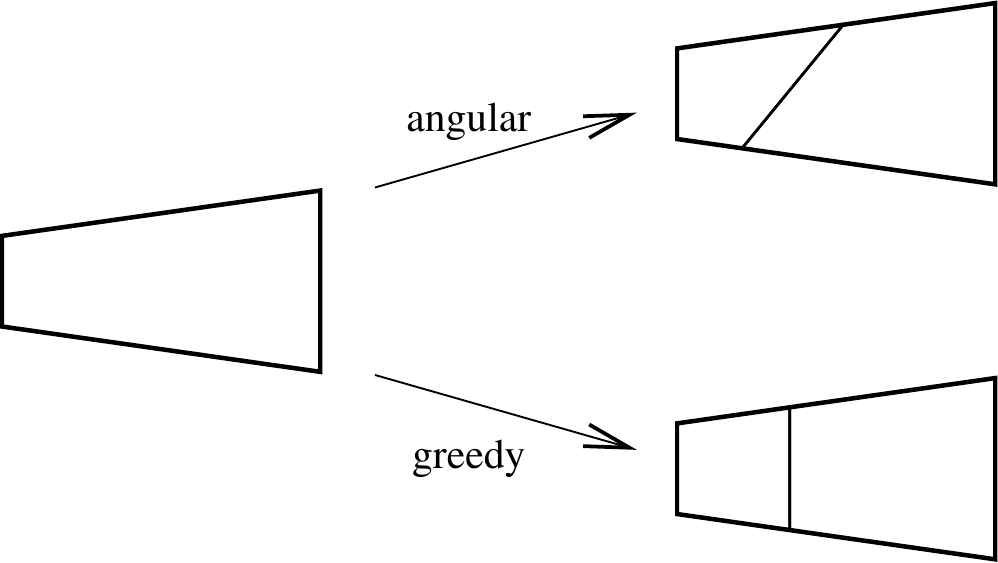}
\end{center}
\caption{Sample executions of our cutting methods.}
\label{fig:sample_cuts}
\end{figure}

\paragraph{Methods for cutting a polygon.}
Suppose we have to cut a convex polygon $P(\node)$ into two subpolygons.
Note that if we fix an orientation for the cut,
then there are only two choices left for the cut because the areas
of the subpolygons are prespecified. (The two choices correspond to having
the smaller of the two areas to the left or to the right of the cut.)
Our cutting methods, depicted in Figure~\ref{fig:sample_cuts} are the following.
\begin{itemize}
\item \emph{Angular}:
Let $c$ denote the cut, that is, the line segment separating the two subpolygons of $P(\node)$.
We select the orientation of $c$ such that we maximize
$$\min\{\myangle(c,e):\hbox{$e$ is an edge of the input polygon}\},$$
where $\myangle(c,e)$ is the smaller of the angles between the lines
$\ell(c)$ and $\ell(e)$ containing $c$ and $e$, respectively.
In other words, we cut in a direction as different as possible
from all the orientations determined by the edges of the polygon.
We then take any of the two cuts of the selected orientation.

\item \emph{Greedy}: The greedy method selects the cut that minimizes
the maximum of the aspect ratios of the two subpolygons.
\end{itemize}

\paragraph{Experiments.}
To get an idea of the relative performance of the two methods
we implemented them and performed some experiments. Table~\ref{table:experiments}
shows the results on two hierarchies.
One is synthetic and was generated using a random process,
the other is the home directory with all subfolders of one of the authors.
Leaves in the latter hierarchy are the files in any of the folders, and the weight
of a leaf is the size of the corresponding file.
For comparison, we also ran two additional partitioning methods.
Both of these methods use the transformation of the input into
a binary tree. The \emph{random} method always makes a cut in a random direction.
In the \emph{greedy rectangular} method, all polygons are rectangles and
all cuts are parallel to the sides of the original rectangle.
The method always greedily chooses the cut that is perpendicular
to the longer side, which maximizes the aspect ratio of the two subpolygons resulting
from the cut.

\ifCtableInstalled
  \begin{table}
  \begin{center}
  \begin{tabular}{!{\vrule width 2pt}c!{\vrule width 2pt}c|c|c|c!{\vrule width 2pt}}
  \noalign{\hrule height 2pt}
  \multirow{2}*{\bf Method} & \multicolumn{2}{c|}{\bf Synthetic Data} &  \multicolumn{2}{c!{\vrule width 2pt}}{\bf Home Folder}\\
  \cline{2-5}
  & \bf Average & \bf Maximum &  \bf Average & \bf Maximum\\
  \noalign{\hrule height 2pt}
  Angular & 3.79 & 14.28 & 3.87 & 20.11\\
  \hline
  Greedy & 2.55 & 6.79 & 2.53 & 10.42\\
  \hline
  Random & 5.75 & 700.41 & 5.99 & 1609.66\\
  \hline
  Greedy Rectangular& 4.25 & 8351.42 & 110.75 & 622590.66\\
  \noalign{\hrule height 2pt}
  \end{tabular}
  \end{center}
  \caption{Aspect ratios of partitions generated by various methods}
  \label{table:experiments}
  \end{table}
  }
\else
  \begin{table}
  \begin{center}
  \begin{tabular}{|c|c|c|c|c|}
  \hline
  {\bf Method} & \multicolumn{2}{c|}{\bf Synthetic Data} &  \multicolumn{2}{c|}{\bf Home Folder}\\
  \cline{2-5}
  & \bf Average & \bf Maximum &  \bf Average & \bf Maximum\\
  \hline
  Angular & 3.79 & 13.19 & 3.87 & 20.11\\
  \hline
  Greedy & 2.56 & 6.79 & 2.57 & 8.39\\
  \hline
  Random & 5.79 & 355.14 & 6.26 & 1609.66\\
  \hline
  Greedy Rectangular& 3.52 & 1445.99 & 24.49 & 230308.30\\
  \hline
  \end{tabular}
  \end{center}
  \caption{Aspect ratios of partitions generated by various methods}
  \label{table:experiments}
  \end{table}
\fi

In all our tests, the methods partitioned a square.
The greedy method performs best, closely followed by the angular method.
Interestingly, the greedy rectangular method performs even worse
than the random method---apparently restricting to rectangles
is a very bad idea as far as aspect ratio is concerned.

In the next two sections we prove that both the angular method and the
greedy method construct a partition in which the aspect ratios are
$O(\poly(\mydepth{\tree} + \log n))$. For the angular method, the proof is simpler and gives
a better bound on the worst-case aspect ratio. We present the more complicated proof
for the greedy method because it is the most natural method and it has the
best performance in practice.

\section{Analysis of angular partitions}
\label{sect:angular_analysis}
The idea behind the angular partitioning method is that a polygon with large aspect
ratio must have two edges that are almost parallel. Hence, if we avoid using partition
lines whose orientations are too close to each other, then we can control the
aspect ratio of our subpolygons.
Next we make this idea precise.

Let $U$ be the initial unit square that we partition, and let $\phi>0$ be a parameter.
Recall that for two line segments $e$ and $e'$, we use $\myangle(e,e')$ to denote
the smaller angle defined by the lines $\ell(e)$ and $\ell(e')$
containing $e$ and $e'$, respectively.
We define a convex polygon $P\subset U$ to be a
\emph{$\phi$-separated polygon} if it satisfies the following condition.
For any two distinct edges $e$ and $e'$, we have:
\begin{enumerate}
\item[(i)] $\myangle(e,e') \geq \phi$; or
\item[(ii)] $e$ is contained in $U$'s top edge and $e'$
           is contained in $U$'s bottom edge (or vice versa); or
\item[(iii)] $e$ is contained in $U$'s left edge and $e'$
           is contained in $U$'s right edge (or vice versa).
\end{enumerate}

\begin{lemma}
\label{le:aspect-ratio}
The aspect ratio of a $\phi$-separated polygon $P$ is $O(1/\phi)$.
\end{lemma}
\begin{proof}
Let $d := \diam(P)$ and let
$uv$ be a diagonal of $P$ that has length~$d$.
Consider the bounding box~$B$ of~$P$ that has two
edges parallel to $uv$. We call the edge of $B$ parallel to and above $uv$ its
\emph{top edge}, and the edge of $B$ parallel to and below $uv$ its \emph{bottom edge}.
Let $r$ be a vertex of $P$ on the top edge of $B$ and let $s$ be a vertex on its bottom
edge---see Fig.~\ref{fi:aspect-ratio}.
\begin{figure}[ht]
\begin{center}
\includegraphics[scale=0.8]{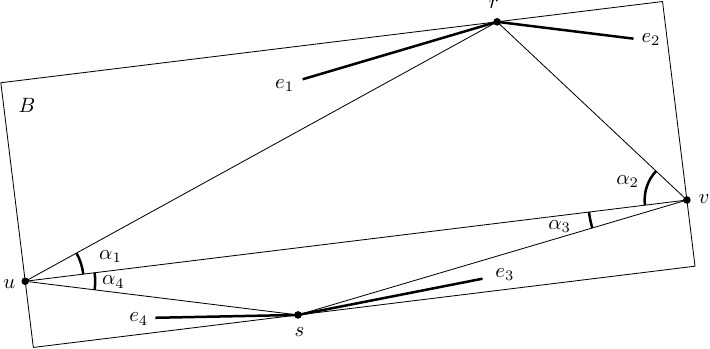}
\caption{Illustration for the proof of Lemma~\ref{le:aspect-ratio}.}
\label{fi:aspect-ratio}
\end{center}
\end{figure}
Let $e_1$ and $e_2$ be the edges of $P$ incident to $r$, and let $e_3$ and $e_4$
be the edges incident to $s$. Let the angles $\alpha_1,\ldots,\alpha_4$ be defined as in Fig.~\ref{fi:aspect-ratio}.
We distinguish two cases.
\begin{itemize}
\item \emph{Case (a): none of $e_1,e_2,e_3,e_4$ are parallel}. \\
  By condition~(i), this implies that the angles any two edges make is at least~$\phi$.
  Hence, we have
  \begin{align}
  \phi &\leq \myangle(e_1,e_3) \leq \max(\alpha_1,\alpha_3), \label{eq1} \\
  \phi &\leq \myangle(e_2,e_4) \leq \max(\alpha_2,\alpha_4), \label{eq2} \\
  \phi &\leq \myangle(e_1,e_4) \leq \alpha_1 + \alpha_4,     \label{eq3} \\
  \phi &\leq \myangle(e_2,e_3) \leq \alpha_2 + \alpha_3.     \label{eq4}
  \end{align}
  By (\ref{eq3}) we have $\alpha_1+\alpha_4 \geq \phi$.
  Now assume without loss of generality that $\alpha_1 \geq \phi/2$.
  If $\alpha_2 \geq \phi/2$ as well, then
  \[
  \area(uvr) \geq (d^2/4)\cdot \sin (\phi/2).
  \]
  Since $uvr\subset P$, this implies that
  \[
  \ar(P) \leq \frac{d^2}{(d^2/4)\cdot \sin (\phi/2)} =  O(1/\phi).
  \]
  If $\alpha_2< \phi/2$, then we use (\ref{eq2}) and (\ref{eq4})
  to conclude that $\alpha_4 \geq \phi$ and $\alpha_3 \geq \phi/2$.
  Hence, we now have $\area(uvs) \geq (d^2/4)\cdot \sin (\phi/2)$,
  which implies that $\ar(P)=O(1/\phi)$.
\item \emph{Case (b): some edges in $e_1,e_2,e_3,e_4$ are parallel}. \\
  By conditions (i)--(iii), two edges of $P$ can be parallel
  only if they are contained in opposite edges of $U$.
  Hence, $|rs| \geq 1$. Moreover, since $uv$ defines the diameter we have $|uv|\geq |rs|$.
  Let $\alpha:= \myangle(uv,rs)$, as illustrated in Fig.~\ref{fi:aspect-ratio-parallel}.
  \begin{figure}[ht]
  \begin{center}
  \includegraphics[scale=0.8]{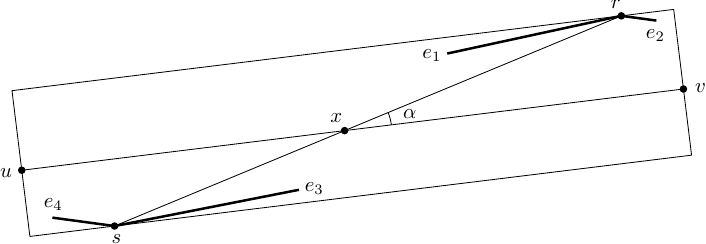}
  \caption{The case of parallel edges.}
  \label{fi:aspect-ratio-parallel}
  \end{center}
  \end{figure}
  If $\alpha\geq \min\{\phi,\pi/4\}$, this is easily
  seen to imply that $\area(P) \geq \area(urvs) = \Omega(\phi)$,
  which means that $\ar(P)=O(1/\phi)$. Now consider the
  case where $\alpha<\phi$ and $\alpha < \pi/4$. We show that this leads to a contradiction.
  Let $x:=uv\cap rs$ and assume without loss of generality that,
  as in Figure~\ref{fi:aspect-ratio-parallel}, we have $\alpha=\angle rxv$.
  Then $\myangle(e_1,e_3) \leq \alpha<\phi$. By conditions~(i)--(iii)
  this can only happen if $e_1$ and $e_3$ are contained in opposite edges of~$U$,
  if we assume that $r$ and $s$ are chosen such that they maximize the distances
  between $r$ and $s$, which can be done without loss of generality.
  We have $\myangle(e_3,rs) \le \myangle(uv,rs) = \alpha \le \phi < \pi/4$.
  However, it is impossible to place two points $r$ and $s$ on opposite sides of $U$
  such that $\myangle(e_3,rs) < \pi/4$. The smallest angle one can obtain is $\pi/4$.
\end{itemize}
\end{proof}

To construct a polygonal partition we use the procedure described in Section~\ref{sect:two_methods}.
Thus, we first transform the input tree $\tree$ into a corresponding binary tree $\tree'$.
Next, we recursively apply the angular cutting method to $\tree'$,
that is, at each node $\node$, we cut the polygon $P(\node)$, using a cut $c$
that maximizes the minimum angle $c$ makes with any of the edges of $P(\node)$.

\begin{lemma}
Let $P(\node)$ be the subpolygon generated by the algorithm above
for a node $\node$ at level $k$ in $\tree'$.
Then $P(\node)$ is a $(\pi/(2k+6))$-separated polygon.
\end{lemma}
\begin{proof}
The proof is by induction on $k$.

For $k=0$, we have $\node=\myroot(\tree')$ and $P(\node)$ is a unit square.
Hence, $P(\node)$ is $(\pi/2)$-separated and therefore also $(\pi/6)$-separated.

For $k>0$ we argue as follows. By the induction hypothesis, the polygon $P(\othernode)$
corresponding to the parent $\othernode$ of $\node$ is $(\pi/(2k+4))$-separated. Moreover,
by construction it has at most $4+(k-1)=k+3$ edges.
Consider the sorted (circular) sequence of angles that these edges make with the $x$-axis.
By the pigeon-hole principle, there must be two adjacent angles that are at least
$\pi/(k+3)$ apart. Hence, the cut $c$ that is chosen to partition $P(\othernode)$
makes an angle at least $\pi/(2k+6)$ with all edges of $P(\othernode)$,
and by the induction hypothesis, all the other angles
are at least $\pi/(2k+4)$.
\end{proof}
Recall that the height of the binary tree $\tree'$ is $O(\mydepth(\tree) +\log n)$.
Hence, we get the following theorem.
\begin{theorem}
Let $\tree$ be a properly weighted tree with $n$ nodes. Then the angular partitioning method
constructs a polygonal partition for $\tree$ whose aspect ratio is $O(\mydepth(\tree)+\log n)$.
\end{theorem}

\section{Analysis of greedy partitions}
\label{sect:greedy_analysis}
We now turn our attention to the greedy method, which at each step
chooses a cut $c$ that minimizes the maximum aspect ratio of the
two subpolygons resulting from the cut.
The main component of our proof that polygonal partitions
with good properties exist will be the following lemma. It shows that there is always a way to cut
a polygon into two smaller polygons of required areas so that
the aspect ratios of the new subpolygons are bounded. 

Note that the lemma requires that the number of vertices in the polygon be bounded.
If the number of vertices is unbounded, the polygon may become arbitrarily close to a circle. In this case there is no good cut if one
of the resulting polygons has to be much smaller than the other.

The proof of the lemma is long and consists of a case analysis. An impatient reader may prefer to omit the proof and move directly to Theorem~\ref{theorem:plane_partition}.

\begin{lemma}[Good cuts]\label{lemma:circular_cut}%
Let $P\subset \mathbb{R}^2$ be a convex polygon with $k$ vertices, and let $a\in (0,1/2]$.
Then $P$ can be partitioned into two convex polygons $P_1$ and $P_2$ such that
\begin{itemize}
\item
  Each of the $P_1$ and $P_2$ has at most $k+1$ vertices.
\item
  $\area(P_1) = a \cdot \area(P)$, and
  $\area(P_2) = (1-a) \cdot \area(P)$.
\item
  $
  \max\{\ar(P_1), \ar(P_2)\} \leq \max\left\{\ar(P)\left(1+\frac{6}{k}\right), k^8\right\}
  $.
\end{itemize}
\end{lemma}

\begin{proof}
We distinguish two cases, depending on whether $a\leq 1/k^2$ (that is,
we are cutting off a relatively small subpolygon) or not.
\begin{description}
\item{Case 1:  $a\leq 1/k^2$.}
Let $\phi$ be the smallest angle of $P$, and let $v$ be a vertex of $P$ whose interior angle is~$\phi$.
Since $P$ has $k$ vertices, we have
\[
\phi \leq \pi \left(1-\frac{2}{k}\right).
\]
Let $\ell$ be the angular bisector at~$v$.
Consider the cut $c$ orthogonal to $\ell$ such that $\area(P_1)=a\cdot\area(P)$,
where $P_1$ is the subpolygon  induced by~$c$ having $v$ as a vertex---see
Figure~\ref{figure:case1-1} for an illustration.
Let $P_2$ be the other subpolygon.
Clearly, $P_1$ and $P_2$ are convex polygons of the required area
with at most $k+1$ vertices each.
Therefore, it remains to bound the aspect ratios of $P_1$ and $P_2$.

Since $P_2\subset P$, we have
\begin{align*}
\ar(P_2) &= \frac{\diam(P_2)^2}{\area(P_2)}
                \leq \frac{\diam(P)^2}{(1-a)\cdot \area(P)}
                =     \frac{\ar(P)}{1-a}
                < \ar(P) \left(1+2a\right) \\
          &< \ar(P) \left(1+\frac{2}{k^2}\right)
                < \ar(P) \left(1+\frac{1}{k}\right).
\end{align*}

We next bound $\ar(P_1)$.  Let $x_1,x_2$ be the two endpoints of the cut~$c$,
and let $t$ be the distance between $x_1$ and $x_2$.
Let $h$ be the distance from $v$ to $c$.
We distinguish between two subcases.
\begin{figure}
\begin{center}
\subfigure[Case 1.\label{figure:case1-1}]{
\scalebox{1}{\includegraphics{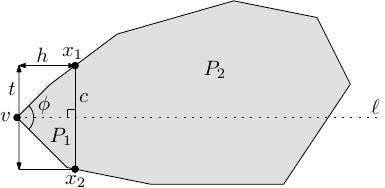}}
}
\hspace{1cm}
\subfigure[Case 1.2.\label{figure:case1-2}]{
\scalebox{1}{\includegraphics{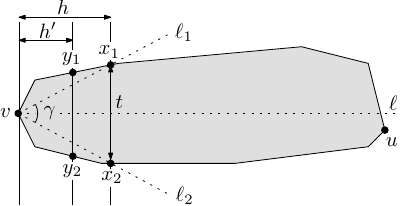}}
}
\caption{Partitioning $P$ into $P_1$ and $P_2$ when $a \leq 1/k^2$.}
\end{center}
\end{figure}
\begin{description}
\item{Case 1.1:  $t \geq h/k^2$.}
Since $P$ is convex, the triangle $vx_1x_2$ is contained in $P_1$.
Therefore,
\[
\area(P_1)\geq h\cdot t/2 \geq h^2/(2k^2).
\]
On the other hand, since $c$ is normal to the bisector of the angle of $v$, 
it follows that $P_1$ is contained inside a rectangle of width $h$ and height $H$, with
\begin{align*}
H &\leq 2\cdot h\cdot \tan(\phi/2)
   \leq 2\cdot h \cdot \tan\left(\frac{\pi(1-2/k)}{2}\right)\\
   &\leq 2\cdot h / \tan(\pi/k)
   \leq  2\cdot h \cdot k / \pi.
\end{align*}

Thus, $\diam(P_1) < h(1+2\cdot k/\pi)$.
It follows that
\[
\ar(P_1) = \frac{\diam(P_1)^2}{\area(P_1)} < \frac{(h+2\cdot h\cdot k/\pi)^2}{h^2/(2k^2)} < k^5.
\]
\item{Case 1.2:  $t < h/k^2$.}
Let $\ell_1$ be the line passing through $v$ and $x_1$,
and let $\ell_2$ be the line passing through $v$ and $x_2$.
Let $\gamma$ be the angle between $\ell_1$ and $\ell_2$.
Observe that $P_2$ is contained between $\ell_1$ and $\ell_2$.
Therefore, if $u$ is the point in $P_2$ farthest away from $v$ we have
\[
\frac{\gamma}{2\pi}\pi \dist{u}{v}^2 \geq \area(P_2),
\]
where $\dist{u}{v}$ denotes the length of the segment~$uv$.
It follows that
\[
\diam(P)^2 \geq \dist{u}{v}^2 \geq \frac{2}{\gamma} (1-a)\cdot\area(P).
\]
Therefore,
\[
\ar(P) = \frac{\diam(P)^2}{\area(P)} \geq \frac{2}{\gamma}\left(1-a\right)
\geq \frac{2}{\gamma}\left(1-\frac{1}{k^2}\right).
\]

We now give an upper bound on the diameter of $P_1$.
Assume without loss of generality that $\dist{v}{x_2} \geq \dist{v}{x_1}$.
Consider a segment $y_1 y_2$ parallel to $x_1 x_2$ with $y_1,y_2\in\bd P$
such that $y_1 y_2$ lies between $x_1 x_2$ and $v$---see Figure \ref{figure:case1-2}.
Let $h'$ be the distance between $v$ and $y_1 y_2$.
We first argue that $\dist{y_1}{y_2}\leq 2t$.

Assume for the sake of contradiction that $\dist{y_1}{y_2} > 2t$.
Let $g_1$ be the line passing through $y_1$ and $x_1$, and let $g_2$ be the line 
passing through $y_2$ and $x_2$.
Observe that since $\dist{y_1}{y_2} > \dist{x_1}{x_2}$, the lines $g_1$ and $g_2$ 
intersect in a point $w$ such that $P_2$ is contained in the triangle~$x_1 x_2 w$.
Furthermore, the polygon $v y_1 x_1 x_2 y_2$ is contained in $P_1$.
If $h'\geq h/2$, then the area of the triangle $v y_1 y_2$ is greater or 
equal to the area of the triangle $x_1 x_2 w$.  Therefore, $\area(P_1)\geq \area(P_2)$, 
contradicting the fact that $a\leq 1/k^2$.
If, on the other hand, $h'<h/2$, then the area of the quadrilateral $y_1 x_1 x_2 y_2$, 
is greater than the area of the triangle $x_1 x_2 w$, again implying that $\area(P_1)\geq \area(P_2)$, 
a contradiction. Therefore, we obtain that $\dist{y_1}{y_2} \leq 2t$.

It now follows that any point $q\in P_1$ is at distance at most $2t$ from the line segment $vx_2$.
Moreover, we have $t<h/k^2\leq \dist{v}{x_2}/k^2$.
Hence,
\[
\diam(P_1)
  = \max_{q,q'\in P_1}\dist{q}{q'}
  \leq 2t + \dist{v}{x_2} + 2t
  \leq \dist{v}{x_2}\left(1+\frac{4}{k^2}\right).
\]

Let $x^*$ be the point on the line segment $x_1x_2$ that is closest to $v$.
Since $\dist{v}{x_2}\geq h$, we have
\[
\area(P_1) \geq \frac{\gamma}{2\pi} \pi \dist{v}{x^*}^2
           \geq \frac{\gamma}{2} (\dist{v}{x_2}-t)^2
           \geq \frac{\gamma}{2} \dist{v}{x_2}^2 \left(1-\frac{1}{k^2}\right).
\]
Therefore,
\begin{align*}
\ar(P_1) &= \frac{\diam(P_1)^2}{\area(P_1)}
         \leq \frac{2}{\gamma} \cdot \frac{(1+4/k^2)^2}{1-1/k^2}
         \leq \ar(P) \frac{(1+4/k^2)^2}{(1-1/k^2)^2} \\
         &\leq \ar(P) \cdot (1+6/k^2)^2
         \leq \ar(P) \cdot (1+2/k)^2\\
         &\leq \ar(P) \cdot (1+6/k).
\end{align*}
\end{description}
\item{Case 2:  $a>1/k^2$.}
\begin{description}
\item{Case 2.1:  $\ar(P) \leq k^6$.}
In this case any cut giving the two subpolygons $P_1$ and $P_2$ the required
areas works. Indeed,
\[
\ar(P_1) = \frac{\diam(P_1)^2}{\area(P_1)}
         \leq \frac{\diam(P)^2}{a\cdot \area(P)}
         \leq k^2 \cdot \ar(P) \leq k^8
\]
and
\[
\ar(P_2) = \frac{\diam(P_2)^2}{\area(P_2)}
         \leq \frac{\diam(P)^2}{(1-a)\cdot \area(P)}
         \leq 2\cdot \ar(P)
         \leq 2 \cdot k^6
         < k^7.
\]

\item{Case 2.2:  $\ar(P) > k^6$.}
Pick points $v_1,v_2\in P$, such that $\dist{v_1}{v_2} = \diam(P)$.
For each $z\in [0,\diam(P)]$, let $\ell(z)$ be a line normal to $v_1v_2$
that is at distance $z$ from $v_1$ and intersects~$P$.
Note that $\ell(0)$ contains $v_1$ and $\ell(\diam(P))$ contains~$v_2$.
Define $f(z)$ to be the length of the intersection of $P$ with $\ell(z)$.
Observe that
\[
\area(P) = \int_{z=0}^{\diam(P)} f(z) dz
\]
Pick $s_1, s_2\in [0,\diam(P)]$, so that
\[
a\cdot \area(P) = \int_{z=0}^{s_1} f(z) dz = \int_{z=\diam(P)-s_2}^{\diam(P)} f(z) dz.
\]
Let $Q_1$ be the part of $P$ that is contained between $\ell(0)$ and $\ell(s_1)$.
Similarly, let $Q_2$ be the part of $P$ that is contained between $\ell(\diam(P)-s_2)$ and $\ell(\diam(P))$.
Clearly, both $Q_1$ and $Q_2$ are convex polygons with at most $k+1$ vertices.
\begin{figure}
\begin{center}
\scalebox{1}{\includegraphics{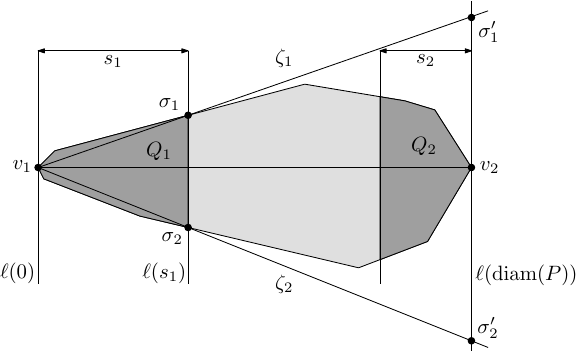}}
\caption{Partitioning $P$ into $P_1$ and $P_2$, when $\alpha > 1/k^2$:  Case~2.2.\label{figure:case2-2}}
\end{center}
\end{figure}

First, we will show that
\[
\min\left\{\frac{\area(Q_1)}{s_1}, \frac{\area(Q_2)}{s_2}\right\}\leq \frac{\area(P)}{\diam(P)}
\]
Assume for a contradiction that both
$\area(Q_1)/s_1$ and 
$\area(Q_2)/s_2$ are greater than
$\area(P)/\diam(P)$.
It follows that there exist $z_1\in [0,s_1]$ and $z_2\in [\diam(P)-s_2]$
such that $f(z_1) > \area(P)/\diam(P)$ and
$f(z_2) > \area(P)/\diam(P)$.
Since $P$ is convex, $f$ is a bitonic function.
Therefore, for each $z\in [z_1,z_2]$, we have $f(z) > \area(P)/\diam(P)$.
It follows that
\[
\area(P) = \area(Q_1) + \area(Q_2) + \area(P\setminus (Q_1\cup Q_2))
         > \frac{\area(P)}{\diam(P)} \cdot \diam(P),
\]
a contradiction.

We can therefore assume without loss of generality that
\[
\frac{\area(Q_1)}{s_1} \leq \frac{\area(P)}{\diam(P)}.
\]
Note that this implies
\[
s_1 \geq a\cdot \diam(P).
\]
We set $P_1=Q_1$, and $P_2 = P\setminus Q_1$.
It remains to bound $\ar(P_1)$ and $\ar(P_2)$.

By the convexity of $P$, we have
\[
\area(P) \geq \max_{z\in [0,\diam(P)]} f(z) \cdot \diam(P)/2.
\]
Since $\ar(P)>k^6$, it follows that
\[
\max_{z\in [0,\diam(P)]} f(z)
  \leq 2 \cdot \frac{\area(P)}{\diam(P)^2}\cdot \diam(P)
  < \frac{2}{k^6}\cdot \diam(P).
\]
This implies that $P$ is contained inside a rectangle with one edge of length $\diam(P)$ parallel to $v_1v_2$, and one edge of length $\frac{4}{k^6}\cdot \diam(P)$ normal to $v_1v_2$.
Thus,
\[
\diam(P_1) \leq s_1 + \frac{4}{k^6}\cdot \diam(P).
\]
Let $\sigma_1$, $\sigma_2$ be the two points where $\ell(s_1)$ intersects $\bd P$.
Let $\zeta_1$, $\zeta_2$, be the lines passing through $v_1$ and $\sigma_1$,
and $v_1$ and $\sigma_2$, respectively.
Let also $\sigma_1'$ and $\sigma_2'$ be the points where $\zeta_1$ and $\zeta_2$, respectively, intersect~$\ell(\diam(P))$---see Figure \ref{figure:case2-2}.

By the convexity of $P$ and $P_1$, we have
\[
\area(P_1) \geq \area(v_1 \sigma_1 \sigma_2)
            =   \left(\frac{s_1}{\diam(P)}\right)^2 \cdot \area(v_1 \sigma_1' \sigma_2')
           \geq \left(\frac{s_1}{\diam(P)}\right)^2 \cdot \area(P).
\]

Since $\area(P_1)= a \cdot \area(P)$, it follows that
$
s_1\leq \sqrt{a} \cdot \diam(P).
$
Using that $a>1/k^2$ we can now derive
\begin{align*}
\ar(P_1)
  &=  \frac{\diam(P_1)^2}{\area(P_1)}
   \leq \frac{(s_1 + 4\cdot \diam(P) / k^6)^2}{\area(P_1)}\\
  &\leq \frac{(\sqrt{a} \cdot \diam(P) + 4\cdot \diam(P) / k^6)^2}{a\cdot \area(P)}
  \\
  &< \frac{\diam(P)^2}{\area(P)} \cdot \left(1+\frac{4}{k^6 \sqrt{a}}\right)^2
  \leq  \ar(P) \cdot \left( 1+\frac{8}{k^4} + \frac{16}{k^{16}} \right)
  \\
  &\leq \ar(P) \cdot \left( 1+\frac{1}{k} \right).
\end{align*}

Since $f$ is bitonic, it follows that
\[
\min_{z\in [s_1,\diam(P)-s_2]}f(z) \geq \min \left\{ \max_{z\in [0,s_1]}f(z), \max_{z\in [\diam(P)-s_2,\diam(P)]}f(z) \right\}.
\]
Therefore,
\[
\frac{\area(P_2)}{\diam(P)-s_1} \geq \frac{\area(P_1)}{s_1}.
\]
Because $P_2$ is contained in a rectangle with one edge of length
$\diam(P)-s_1$ parallel to $v_1v_2$, and one edge of length $\frac{4}{k^6}\cdot \diam(P)$
normal to $v_1v_2$, we have
\[
\diam(P_2) \leq \diam(P) - s_1 + \frac{4}{k^6}\cdot \diam(P).
\]
Putting everything together, we get
\begin{align*}
\ar(P_2) &= \frac{\diam(P_2)^2}{\area(P_2)}
         \leq \frac{(\diam(P)\cdot(1+4/k^6) - s_1)^2}{(1-a)\cdot \area(P)}\\
         &\leq \ar(P) \cdot \left(\frac{1+4/k^6-a}{\sqrt{1-a}}\right)^2
         \leq \ar(P) \cdot \left( 1+4\cdot \frac{\sqrt{2}}{k^6} \right)^2\\
         &\leq \ar(P) \cdot \left( 1+ \frac{1}{k^2} \right)^2
         \leq \ar(P) \cdot \left( 1+ \frac{3}{k^2} \right) \\
         &\leq \ar(P) \cdot \left( 1+ \frac{1}{k} \right).
\end{align*}
\end{description}
\end{description}
This concludes the proof.
\end{proof}
Now we have all the necessary tools to prove a bound on the
aspect ratio of the polygonal partition constructed by
the greedy method.
\begin{theorem}%\textsc{(Existence of Hierarchical Circular Partitions)}
\label{theorem:plane_partition}
Let $\tree$ be a properly weighted tree with $n$ nodes.
Then the greedy partitioning method constructs a polygonal partition for $\tree$
whose aspect ratio is $O\left(\left(\mydepth(\tree) + \log{n}\right)^8\right)$.
\end{theorem}
\begin{proof}
Recall from Section~\ref{sect:two_methods} that our algorithm starts by transforming
$\tree$ to a binary tree~$\tree'$ of height $O(\mydepth(\tree)+\log n)$.
This is done in such a way that a polygonal partition for
$\tree'$ induces a polygonal partition for $\tree$ of the same (or better) aspect ratio.
We then recursively apply the greedy cutting strategy to $\tree'$. Hence, it suffices to
show that a recursive application of the greedy strategy to a binary tree of height~$h$
produces a polygonal partition of aspect ratio $(h+3)^{\gb}$.

Let $A(i)$ be the worst-case aspect ratio of a polygon $P(\node)$ produced
by the greedy method over all nodes $\node$ at depth~$i$ in the tree.
Note that $P(\myroot(\tree))$ is a square and each polygon
at depth~$i$ has at most $i+4$ vertices.
By Lemma~\ref{lemma:circular_cut} we thus have
\[
A(i) \leq \begin{cases}
                 2 & \mbox{if $i=0$,} \\
                 \max \left\{ (i+3)^8, \left(1+\frac{6}{i+3}\right)\cdot A(i-1) \right\} & \mbox{if $i>0$.}
              \end{cases}
\]
We can now prove by induction that $A(i) \leq (i+3)^8$. Indeed, for $i=0$ this
obviously holds, and for $i>0$ we have
\begin{align*}
A(i) &\leq \max \left\{ (i+3)^8, \left(1+\frac{6}{i+3}\right)\cdot A(i-1) \right\}\\
&\leq \max \left\{ (i+3)^8, \left(1+\frac{6}{i+3}\right)\cdot (i+2)^8 \right\}
\end{align*}
and
\begin{align*}\textstyle
\left(1+\frac{6}{i+3}\right)\cdot (i+2)^8
  &= \textstyle\left\{ \left(1+\frac{6}{i+3}\right) \cdot \left( \frac{i+2}{i+3} \right)^8 \right\} \cdot (i+3)^8 \\
  &= \textstyle\left\{ \frac{1+\frac{6}{i+3}}{\left(1+\frac{1}{i+2}\right)^8} \right\} \cdot (i+3)^8
  \\
  &< \textstyle\left\{ \frac{1+\frac{6}{i+3}}{1+\frac{8}{i+2}} \right\} \cdot (i+3)^8
  \\
  &< \textstyle(i+3)^8.
\end{align*}
Hence, $A(h)< (h+3)^8$, which finishes the proof.
\end{proof}

\section{A lower bound for polygonal partitions}\label{sect:lower-bound}
In the previous sections we have seen that any properly weighted tree $\tree$
with $n$ nodes admits a polygonal partition whose aspect ratio is
$O(\mydepth(\tree)+\log n)$. In this section we prove this is almost tight
in the worst case, by exhibiting a tree for which any polygonal partition
has aspect ratio $\Omega(\mydepth(\tree))$.

We start with an easy lower bound on the aspect ratio of a convex polygon
in terms of its smallest angle.
\begin{observation}\label{le:angle-vs-aspect}
Let $P$ be a convex polygon and let $\alpha$ be the smallest interior angle of~$P$.
Then the aspect ratio of $P$ is at least $2/\alpha$.
\end{observation}
\begin{proof}
Let $v$ be a vertex of $P$ whose interior angle is $\alpha$.
Then $P$ is contained in the circular sector with radius
$\diam(P)$ and angle $\alpha$ whose apex is at $v$.
This sector has area $(\alpha/2)\cdot\diam(P)^2$, from which
the observation readily follows.
\end{proof}
Next we show how to construct, for any given height $h$,
a weighted tree $\tree$ of height $h$ such that any polygonal
partition for $\tree$ has a region with a very small angle.
The lower bound on the aspect ratio then follows
immediately from Observation~\ref{le:angle-vs-aspect}.

\begin{figure}[t]
\begin{center}
\includegraphics{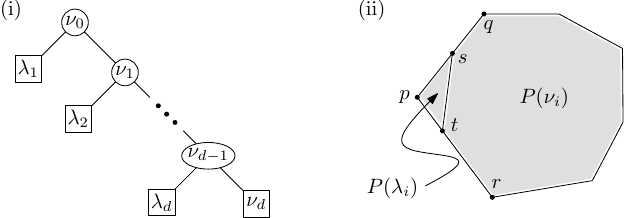}
\caption{(i) Structure of the tree for the lower-bound construction.
         (ii) Illustration for the proof of Lemma~\ref{le:triangle-cutoff}. }
\label{fi:lb-tree}
\end{center}
\end{figure}
The structure of $\tree$ is depicted in Fig.~\ref{fi:lb-tree}(i).
The tree $\tree$ has $h+1$ nodes $\nu_0,\ldots,\node_{h}$ that form a path,
and $h$ other (leaf) nodes $\lambda_1,\ldots,\lambda_{h}$ branching off the path.
The idea will be to choose the weights of the leaves very small, so
that the only way to give the region $P(\lambda_i)$ the required
area, is to cut off a small triangle from~$P(\node_i)$. Then we will argue
that one of these triangles must have a small angle.

To make this idea precise we define
\[
x_i := \begin{cases}
                     1 & \mbox{if $i=0$,} \\
                     x_{i-1} / (2\sqrt{h}) & \mbox{if $0<i\leq h$,}
                  \end{cases}
\]
and we set $\weight(\lambda_i) := x_{i-1}^2 /(4h)$ for $1\leq i\leq h$.
Note that defining the weights for $\lambda_1,\ldots,\lambda_h$
implicitly defines the weights for $\node_1,\ldots,\node_{h}$ as well.

\begin{lemma}\label{le:triangle-cutoff}
If each region created in a polygonal partition for $\tree$ has aspect ratio at most $h$,
then we have for $0\leq i\leq h$,
\begin{enumerate}
\item[(i)] $P(\lambda_{i-1})$ is a triangle (if $\lambda_{i-1}$ exists, that is, if $i\neq 0$);
\item[(ii)] $P(\node_i)$ has $i+4$ sides, each of length at least $x_i$.
\end{enumerate}
\end{lemma}
\begin{proof}
We will prove the lemma by induction on~$i$.
For $i=0$ the lemma is obviously true, since $P(\node_0)$ is the unit square.

Now let $i>0$. By the induction hypothesis, each side of $P(\node_{i-1})$
has length at least~$x_{i-1}$. If $P(\lambda_{i})$ fully contained an edge
of $P(\node_{i-1})$, its diameter would therefore be more than $x_{i-1}$.
But then its aspect ratio would be
\[
\frac{\diam(P(\lambda_{i}))^2}{\area(P(\lambda_{i}))}
  \geq   \frac{x_{i-1}^2}{\weight(\lambda_{i})}
    =    \frac{x_{i-1}^2}{x_{i-1}^2/4h}
    =    4h,
\]
which contradicts the assumptions. Hence, $P(\lambda_{i})$ is a triangle, as claimed.
It also follows that $P(\node_{i})$ has $i+4$ sides. It remains to show that
these sides have length at least~$x_i$.
\medskip

Let $st$ be the segment that cuts $P(\node_{i-1})$ into $P(\node_i)$ and $P(\lambda_{i})$,
let $p$ be the corner of $P(\node_{i-1})$ that is cut off,
and let $q$ and $r$ be the corners of $P(\node_{i-1})$ adjacent
to $p$---see Fig.~\ref{fi:lb-tree}(ii).
All sides of $P(\node_i)$ except $qs$, $st$, and $tr$ have length at least $x_{i-1}$
so they definitely have length at least $x_i$.

It is easy to see that the angles
of any polygon $P(\node_j)$ are at least $\pi/2$---indeed, when a corner is cut
off from some $P(\node_j)$, the two new angles appearing in $P(\node_{j+1})$
are larger than the angle at the corner that is cut off.
Hence, $st$ is the longest edge of the triangle $pst = P(\lambda_{i})$,
and we have
\[
|st| \geq \sqrt{\area(P(\lambda_{i}))}
      = \sqrt{\weight(\lambda_{i})}
      = \sqrt{\frac{x_{i-1}^2}{4h}}
      = \frac{x_{i-1}}{2\sqrt{h}}
      = x_i.
\]
Next we show that $qs$ has length at least $x_i$; the argument
for $tr$ is similar. Note that $|ps| \leq \sqrt{h \cdot \weight(\lambda_{i})}$,
otherwise $P(\lambda_i)$'s aspect ratio would be larger than~$h$.
Hence, we have
\begin{align*}
|qs| &=    |pq| - |ps| \\
     &\geq  x_{i-1} - |ps| \qquad \mbox{(induction hypothesis)} \\
     &\geq  x_{i-1} - \sqrt{h \cdot \weight(\lambda_{i})} \\
     &=  x_{i-1} - \sqrt{h \cdot \frac{x_{i-1}^2}{4h}} \\
     &=  \frac{x_{i-1}}{2} \\
     &\geq  x_i.%
\end{align*}%
\end{proof}
Next we show that one of the triangles $P(\lambda_i)$ must have large aspect ratio.
\begin{lemma}\label{le:small-angle}
If each region created in a polygonal partition for $\tree$ has aspect ratio at most $h$,
then there is a region $P(\lambda_i)$ where one of whose interior angles
is at most $2\pi/(h+4)$.
\end{lemma}
\begin{proof}
By the previous lemma, the region $P(\node_h)$ has $h+4$ sides.
The sum of the interior angles of a $(h+4)$-gon is
exactly $(h+2)\cdot 2\pi$, so one of the interior angles of $P(\node_h)$ must
be at least
\[
\frac{(h+2)\cdot 2\pi}{h+4} = 2\pi - \frac{4\pi}{h+4}.
\]
If the two edges meeting at some corner $p$ of $P(\node_h)$ make an angle of
at least $2\pi - 4\pi/(h+4)$ in $P(\node_h)$, then they must make an
angle of at most $4\pi/(h+4)$ in a region $P(\lambda_i)$ adjacent to $p$.
\end{proof}
\begin{theorem}
For any $h$, there is a weighted tree $\tree$ of height $h$ such that any polygonal partition
for $\tree$ has aspect ratio $\Omega(h)$.
\end{theorem}
\begin{proof}
Consider the tree $\tree$ described above. Assume the aspect ratio of the polygonal partition
for $\tree$ is less than $h$---otherwise we are done.
Then by Lemma~\ref{le:small-angle} there is a region
with interior angle  $4\pi/(h+4)$. By Observation~\ref{le:angle-vs-aspect} this region
has aspect ratio at least $(h+4)/2\pi=\Omega(h)$.
\end{proof}

\section{Partitions with slack}\label{sect:slack}
Recall that a rectangular partition is a polygonal partition in which
all polygons are rectangles.
We show that if we allow a small distortion of the areas of the rectangles,
then there exists a rectangular partition of small aspect ratio.
This result can also be obtained for partitions of a hypercube in~$\Reals^d$.
We now define more precisely which type of distortion we allow.

Let $\tree$ be a properly weighted tree.
A \emph{rectangular partition with $\eps$-slack} in $\Reals^d$ for $\tree$ assigns
a $d$-dimensional hyperrectangle $R(\node)$ to each node $\node\in\tree$ such that
\begin{itemize}
\item the hyperrectangle $R(\myroot(\tree))$ is the unit hypercube in $\Reals^d$;
\item for any two nodes $\node,\othernode$ such that $\othernode$ is a child of $\node$ we have
      \[ (1-\eps)\cdot \frac{\Vol(R(\othernode))}{\weight(\othernode)}
      \leq \frac{\Vol(R(\node))}{\weight(\node)}
      \leq \frac{\Vol(R(\othernode))}{\weight(\othernode)}.
      \]
\item for any node $\node$, the hyperrectangles assigned to the children of $\node$
      have pairwise disjoint interiors and are contained in $R(\node)$.
\end{itemize}
Observe that as we go down the tree $\tree$, the volumes of the hyperrectangles
can start to deviate more and more from their weights. However, the relative volumes
of the hyperrectangles of the children of a node $\node$ stay roughly the same and together
they still cover $R(\node)$ almost entirely.

For convenience, we will work with the \emph{rectangular aspect ratio}
of a $d$-dimensional hyperrectangle rather than using the aspect-ratio definition
given earlier. The rectangular aspect ratio $\RAR(R)$ of hyperrectangle
$R$ with side lengths $s_1,s_2,\ldots,s_d$ is defined as
$\RAR(R) := \frac{\max_i s_i}{\min_i s_i}$.
It can easily be shown that for 2-dimensional rectangles, the aspect ratio and the rectangular
aspect ratio are within a constant factor.
We define the rectangular aspect ratio of a rectangular partition as
the maximum rectangular aspect ratio of any of the hyperrectangles in the partition.

Our algorithm to construct a rectangular partition
always cuts perpendicular to the longest side of a hyperrectangle~$R(\node)$.
Since we can shrink each of the hyperrectangles
of the children of $\node$ by a factor $1-\eps$, we have some extra space
to keep the aspect ratios under control. We will prove that this means that
the longest-side-first strategy can ensure a rectangular aspect ratio of~$1/\eps$.
The basic tool is the following lemma.
\begin{lemma}\label{lemma:slack-cut}
Let $0<\eps<1/3$, and let $R$ be a hyperrectangle with $\RAR(R)\leq 1/\eps$.
Let $S=\{w_1,\ldots,w_k\}$ be a set of weights with $\sum_{i=1}^k w_i = (1-\eps)\cdot \Vol(R)$.
Then there exists a set $\{R_1,\ldots,R_k\}$ of pairwise disjoint hyperrectangles,
each contained in $R$, such that
for all $1\leq i\leq k$ we have $\Vol(R_i) = w_i$ and $\RAR(R_i) \leq 1/\eps$.
\end{lemma}
\begin{proof}
We will prove this by induction on~$k$. The case $k=1$ is trivial, so now assume~$k>1$.
For a subset $S'\subset S$, we define $w(S') := \sum_{w_i\in S'} w_i$.
Assume without loss of generality that $w_1 = \max_i w_i$. There are two cases.
\begin{itemize}
\item If $w_1 \leq (1-\eps)\cdot w(S)$, then we can split $S$ into subsets $S'$
      and $S''$ such that $w(S') \geq \eps\cdot w(S)$ and $w(S'')\geq \eps\cdot w(S)$.
      Indeed, if $i^*$ is the minimum index such that $\sum_{i=1}^{i^*} w_i \geq \eps \cdot w(S)$,
      then $\sum_{i=1}^{i^*} w_i \leq (1-\eps)\cdot w(S)$ since $\eps<1/3$.
      Hence, setting $S' = \{w_1,\ldots,w_{i^*}\}$ satisfies the condition.

      We now cut $R$ perpendicular to its longest side into hyperrectangles
      $R'$ and $R''$ such that
      \[
      \Vol(R') = \Vol(R) \cdot \frac{w(S')}{w(S)} \ \ \mbox{ and } \ \
      \Vol(R'') = \Vol(R) \cdot \frac{w(S'')}{w(S)}.
      \]
      Note that this implies that
      \[
      w(S') = \frac{w(S)}{\Vol(R)} \cdot \Vol(R') = (1-\eps)\cdot \Vol(R').
      \]
      Since $w(S') \geq \eps\cdot w(S)$ and the cut is perpendicular to the longest side,
      we have $\RAR(R') \leq 1/\eps$.  Hence, by induction there exists a set
      of pairwise disjoint hyperrectangles, each contained in $R'$, whose
      volumes are the weights in $S'$ and whose aspect ratios are at most $1/\eps$.
      Similarly, there is a collection of suitable hyperrectangles contained
      in $R''$ for the weights in $S''$.
      Hence, we have found a set of hyperrectangles for $S$ satisfying all the conditions.
\item If $w_1>(1-\eps)\cdot w(S)$ we
      split $R$ into $R_1$ and $R'$ by cutting perpendicular to the longest side,
      such that $\Vol(R_1) = w_1$. Since
      \[
      w_1 > (1-\eps)\cdot w(S) = (1-\eps)^2 \cdot \Vol(R) > \eps \cdot \Vol(R),
      \]
      we have $\RAR(R_1)\leq 1/\eps$.
      Furthermore, since $w(S)=(1-\eps)\cdot \Vol(R)$ we certainly have
      $w_1<(1-\eps)\cdot \Vol(R)$, and so
      \[
      \Vol(R') = \Vol(R)-w_1 > \eps\cdot\Vol(R).
      \]
      This implies that $\RAR(R')\leq 1/\eps$.
      Finally, we note that
      \[
      \frac{w(S\setminus\{w_1\})}{\Vol(R')}
         = \frac{w(S)-w_1}{\Vol(R)-w_1}
         < \frac{w(S)}{\Vol(R)}
         = 1 - \eps.
      \]
      Hence, we can slightly shrink $R'$ such that
      $w(S\setminus\{w_1\}) = (1-\eps)\cdot\Vol(R')$, which means we can apply the induction
      hypothesis to obtain a suitable set of hyperrectangles for the weights
      in $S\setminus \{w_1\}$ inside~$R'$. Together with the hyperrectangle
      $R_1$ for $w_1$, this gives us a collection
      of hyperrectangles for the weights in $S$ satisfying all the conditions.%
\end{itemize}%
\end{proof}
We can now prove our final result on partitions with slack.
\begin{theorem}\label{theorem:rect_partition}
Let $0<\epsilon <1/3$, and let $d\geq 2$.
Let $\tree$ be a properly weighted tree. Then there
exists a rectangular partition with $\eps$-slack in $\Reals^d$ for $\tree$
whose aspect ratio is at most $1/\eps$.
\end{theorem}

\begin{proof}
We construct the rectangular partition recursively,
as follows. We start at the root of $\tree$, where we set
$R(\myroot(\tree))$ to the unit hypercube. Note that the rectangular aspect ratio of
a hypercube is~1.
Now, given a hyperrectangle $R(\node)$ of an internal node $\node$ such that the
rectangular aspect ratio is at most $1/\eps$, we will construct
hyperrectangles for the children of~$\node$ of the required aspect ratio.
To this end, for any child $\othernode$ of $\node$ we define
\[
\weight^*(\othernode) := (1-\eps)\cdot \weight(\othernode) \cdot \frac{\Vol(R(\node))}{\weight(\node)}.
\]
Observe that since $\sum_{\othernode}\weight(\othernode) = \weight(\node)$, we have
\[
\sum_{\othernode} \weight^*(\othernode) = (1-\eps)\cdot \Vol(R(\node)).
\]
Hence, we can apply Lemma~\ref{lemma:slack-cut} with
$S=\{\weight^*(\othernode): \mbox{ $\othernode$ is a child of $\node$} \}$
to partition $R(\node)$ into pairwise disjoint hyperrectangles $R_{\othernode}$,
each contained in $R(\node)$ and of aspect ratio at most~$1/\eps$.
The volume of each $R(\othernode)$ will be $\weight^*(\othernode)$,
which is in the allowed range.
Finally, we recurse on each non-leaf child.
After we the recursive process has finished, we have the desired partition.
\end{proof}

\section{Embedding ultrametrics into $\mathbb R^d$}

Before we describe our algorithm for embedding ultrametrics into $\mathbb{R}^d$,
we define \emph{$\alpha$-hi\-er\-ar\-chi\-cal well separated trees},
or $\alpha$-HSTs for short, introduced by Bartal~\cite{Bar1}.
Let $\alpha>1$ be a parameter. An $\alpha$-HST is a rooted tree~$\tree$
with all leaves on the same level, where
each node $\node$ has an associated label $l(\node)$ such that
for any node~$\node$ with parent $\othernode$ we have $l(\othernode)=\alpha\cdot l(\node)$.
The metric space that corresponds to the HST $\tree$ is defined on the leaves of~$\tree$,
and the distance between any two leaves in $\tree$ is equal to the label of the
their lowest common ancestor.

Let $M=(X,D)$ be the given ultrametric.  After scaling $M$, we can assume that the minimum distance is~1
and the diameter is~$\Delta$. For any $\alpha>1$, we can embed $M$ into an $\alpha$-HST,
with distortion $\alpha$~\cite{BartalMendelUltra}.
Given~$M$, we initially compute an embedding of $M$ into a $2$-HST $\tree$ with distortion 2.
Let $M_{\tree}=(X,D_{\tree})$ be the metric space corresponding to $\tree$.
Any embedding of $M_{\tree}$ into $\Reals^d$ with distortion $c'$, yields an embedding of
$M$ into $\mathbb{R}^d$ with distortion~$O(c')$.
It therefore suffices to embed of $M_{\tree}$ into $\Reals^d$.
With a slight abuse of notation we will from now on
denote the leaf of $\tree$ corresponding to a point $x\in X$ simply by~$x$.

\paragraph{A lower bound.}
To prove the approximation ratio of our embedding algorithm, we need a lower
bound on the optimal distortion. We will use the lower bound proved by
B\u{a}doiu~\etal\cite{BadoiuSOCG2006}, which we describe next.

Consider an embedding $\phi$ of $M_{\tree}$ into $\mathbb{R}^d$.
Assume without loss of generality that $\phi$ is non-contracting, that is,
$\phi$ does not make any distances smaller.
For each node $\node\in \tree$ we define a set $A_{\node}\subset \mathbb{R}^d$ as follows.
Let $B_d(r)$ denote the ball in $\Reals^d$ centered at the origin and of radius~$r$.
For a leaf $\node$ of $\tree$, let $A_\node$ be equal to the ball
$B_d(\frac{1}{2})$ translated so that its center is~$\phi(\node)$.
For a non-leaf node~$\node$ with children $\othernode_1,\ldots, \othernode_k$,
let $A_\node$ be the Minkowski sum of $\bigcup_{i=1}^k A_{\othernode_i}$
with $B_d(l(\node))$.
By the non-contraction of $\phi$ it follows that for each pair of nodes
$\node,\node'$ that are on the same level of $\tree$, the regions $A_{\node}$ and
$A_{\node'}$ have disjoint interiors---see Figure \ref{figure:lower} for an example.
Intuitively, the volume $A_{\node}$ is necessary to embed all the points
corresponding to the leaves below $\node$ such that the embedding is non-contracting.
This in turn can be used via an isoperimetric argument to obtain
a lower bound on the maximum distance (and, hence, distortion) between the images of these leaves.

\begin{figure}
\begin{center}
\scalebox{1}{\includegraphics{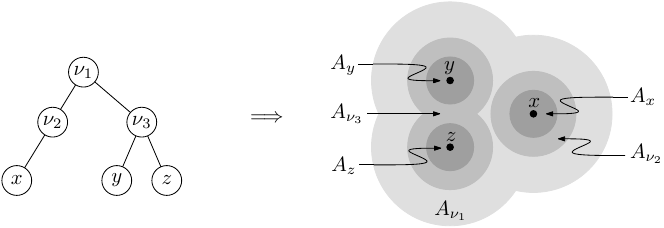}}
\caption{The sets $A_\node$ for a non-contracting embedding of an HST.\label{figure:lower}}
\end{center}
\end{figure}

We cannot compute $A_{\node}$ exactly, however, since we do not know the embedding.
Hence, following B\u{a}doiu~\etal\ \cite{BadoiuSOCG2006} we define for each $\node \in \tree$
a value $A^*(\node)$ that estimates the volume of $A_{\node}$.
Intuitively, the estimate on $A_{\node}$ is derived by the Brunn-Minkowski inequality.
More precisely, we define $A^*(\node)$ as follows.
\[
A^*(\node) = \begin{cases}
                 \Vol\left(B_d(\frac{1}{2})\right) & \mbox{if $\node$ is a leaf,} \\
                 \sum_{\mbox{$\othernode$ is a child of $\node$}} \left( A^*(\othernode)^{1/d}
                                      + \Vol\left(B_d\left(\frac{l(\node)}{4}\right)\right)^{1/d}\right)^d
                       & \mbox{otherwise.}
              \end{cases}
\]

The estimate $A^*(\node)$ can be used to obtain a lower bound on the distortion,
as made precise in the next lemma.
For any $V>0$, let $r_d(V)$ be the radius of a $d$-dimensional ball with
volume~$V$; thus we have 
\[
r_{d}(V)=\left(\frac{V\cdot \Gamma(1+d/2)}{\pi^{d/2}}\right)^{1/d} = \Theta\left(\sqrt{d}\cdot V^{1/d} \right),
\]
where $\Gamma(z) = \int_{0}^{\infty}t^{z-1}e^{-t}dt$.

\begin{lemma}[\cite{BadoiuSOCG2006}, Corollary 1]\label{lemma:ultra_lower}
Let $\opt$ denote the distortion of an optimal embedding of $M_{\tree}$ into $\mathbb{R}^d$.
Then
\[
\opt  \geq 
\max_{\mbox{\rm\small $\node$ is internal node of $\tree$}} \frac{r_d(A^*(\node))}{l(\node)} - 1.
\]
\end{lemma}

\paragraph{The algorithm.}
We are now ready to describe the embedding $f$ of $M_{\tree}$ into $\mathbb{R}^d$.
The intuition behind our algorithm is as follows.
The lower bound given by Lemma \ref{lemma:ultra_lower} implies that an embedding is nearly-optimal
if it results in sets $A_\node$ with small aspect ratio.
Our approach, however, is essentially reversed.  We first compute a hyperrectangular
partition of $\mathbb{R}^d$ into hyperrectangles with small aspect ratio.
The hyperrectangle computed for the leaves in $\tree$ will roughly 
correspond to the balls around the embedded points $x\in X$, so 
given the partition we will be able to obtain the embedding by placing $f(x)$
at the center of the hyperrectangle computed for~$x$. The weights we use for
the nodes of $\tree$ will be the values $A^*(\node)$, which using Lemma~\ref{lemma:ultra_lower}
can then be used to bound the distortion.

More precisely, the algorithm works as follows.
\begin{enumerate}
\item \label{step1} Compute for each node $\node\in\tree$ the value $A^*_\node$.
\item \label{step2} Compute a hyperrectangular partition with slack 
      $\eps:= \min(\frac{1}{3},\frac{1}{\log\Delta})$
      for~$\tree$, where $A^*_\node$ is used as the weight of a node~$\node$.
      Note that $\tree$ is not properly weighted: the weight of $\myroot(\tree)$ will be greater
      than~1 and the weight of an internal node~$\node$ is larger than the sum of the
      weights of its children. However, we can still apply Theorem~\ref{theorem:rect_partition}.
      Indeed, by scaling the weights  appropriately we can ensure that the root
      has weight~1, and the fact that the weight of an internal node~$\node$ is larger than the sum of the
      weights of its children only makes it easier to obtain a small aspect ratio.
      Thus we can compute a hyperrectangular partition for $\tree$ whose
      rectangular aspect ratio is at most $\log \Delta$.
\item \label{step3} Let $P(\node)$ denote the hyperrectangle computed 
      for node $\node\in \tree$ in Step~\ref{step2}.
      We slightly modify the hyperrectangles $P(\node)$, as follows. 
      Starting from the root of $\tree$, we traverse all the nodes of $\tree$.
      When we visit an internal node $\node$, we shrink all hyperrectangles of the nodes
      in the subtree rooted at $\node$ by a factor of $1-1/\log\Delta$, 
      with the center of the (current) hyperrectangle $P'(\node)$ of $\node$ 
      being the fixed point in the transformation.
      Note that $P'(\node)$ itself is not shrunk. Thus the shrinking step moves the
      hyperrectangles contained inside $P'(\node)$ away from its boundary and towards
      its center, thus preventing points in different subtrees to get too close
      to each other.
\item \label{step4} Let $P'(\node)$ denote the hyperrectangle computed in Step~\ref{step3}.
      Then we define the embedding $f(x)$ of a point $x\in X$ to be the center of
      the hyperrectangle~$P'(x)$.
\end{enumerate}
It remains to bound the distortion of $f$. We will need the following observation,
which bounds the diameter of a hyperrectangle in terms of its volume and aspect ratio.
\begin{observation}\label{obs:diam-vs-vol}
Let $P$ be a hyperrectangle in $\Reals^d$ with aspect ratio at most~$\alpha$. 
Then, assuming $\alpha\geq\sqrt{d-1}$,
\[
\diam(P) \leq \alpha\sqrt{2} \cdot \Vol(P)^{1/d}.
\]
\end{observation}
From now on we assume that $\log\Delta\geq \sqrt{d-1}$.
\begin{lemma}\label{lemma:ultra_expansion}
The expansion of $f$ is $O\left(\frac{1}{\sqrt{d}} \cdot \log{\Delta} \cdot \opt\right)$.
\end{lemma}
\begin{proof}
Consider points $x,y\in X$. Let $\node$ be the lowest common ancestor of $x$ and $y$ in~$\tree$.
We have $D_{\tree}(x,y)=l(\node)$.
Both $f(x)$ and $f(y)$ are contained in $P'(\node)$, which in turn is contained in $P(\node)$.
Hence, $\diam(P(\node))$ gives an upper bound on $\dist{x}{y}$.
Note that $\RAR(P(\node)) \leq \log\Delta$ and that $\Vol(P(\node))\leq A^*_{\node}$.
(We do not necessarily have $\Vol(P(\node))= A^*_{\node}$ because of the slack in the partition.)
Hence,
\begin{align*}
\dist{x}{y} 
  &\leq \textstyle\diam(P(\node)) \\
  &\leq \textstyle\log\Delta \cdot \sqrt{2} \cdot \Vol(P(\node))^{1/d} \\
  &\leq \textstyle\log\Delta \cdot \sqrt{2} \cdot \left(A^*_{\node}\right)^{1/d} \\
  &= \textstyle O\left( \log\Delta \cdot r_d(A^*_{\node}) \cdot \frac{1}{\sqrt{d}} \right) \\
  &= \textstyle O\left( \log\Delta \cdot (\opt+1)\cdot l(\node) \cdot \frac{1}{\sqrt{d}} \right) \\
  &= \textstyle O\left( \log\Delta \cdot \opt\cdot D_\tree(x,y) \cdot \frac{1}{\sqrt{d}} \right),
\end{align*}

where the second-to-last transition follows from Lemma~\ref{lemma:ultra_lower}. This shows that the expansion is $O(\log{\Delta} \cdot \opt \cdot \frac{1}{\sqrt{d}})$.
\end{proof}

\begin{lemma}\label{lemma:ultra_contraction}
The contraction of $f$ is $O\left(\sqrt{d}\cdot \log^{2}{\Delta}\right)$.
\end{lemma}
\begin{proof}
Since $\tree$ is a 2-HST, we have
$\Delta=l(\myroot(\tree))=2^{\mydepth(\tree)}$.
Hence,  $\mydepth(\tree)=\log\Delta$. It follows that
for each node $\node\in\tree$,
\begin{align*}
\Vol(P'(\node)) 
  &=  \Omega \left( \left(1-\frac{1}{\log\Delta}\right)^{\log\Delta} \cdot \Vol(P(\node)) \right) 
    =  \Omega \left( \Vol(P(\node)) \right) \\
  &=  \Omega \left( \left(1-\frac{1}{\log\Delta}\right)^{\log\Delta} \cdot A^*_{\node} \right) 
    =  \Omega \left( A^*_{\node} \right).
\end{align*}
Consider points $x,y\in X$, and let $\node$ be the lowest common ancestor of $x$ and $y$ in~$\tree$.
We will consider the following two cases for $\node$:
\begin{itemize}
\item \emph{Case 1:  $\node$ is the parent of $x$ and $y$ in $\tree$.} \\
  Since the minimum distance in $M_{\tree}$ is~1, it follows that $D_{\tree}(x,y)=1$.
  By construction, $f(x)$ is the center of $P'(x)$.
  Let $t$ be the distance between $f(x)$ and $\bd P'(x)$.
  Since $\RAR(P'(x))\leq 1/\eps$, we have
  \[
  t \geq \eps \cdot \Vol(P'(x))^{1/d} 
      =  \Omega\left( \eps \cdot (A^*_x)^{1/d} \right)
      =  \Omega\left(\eps \cdot \frac{1}{\sqrt{d}}\right)
      =  \Omega\left(\frac{1}{\sqrt{d}\log\Delta}\right).
  \]
  Thus,
  $\dist{f(x)}{f(y)} \geq  t = \Omega\left(\frac{D(x,y)}{\sqrt{d}\log\Delta}\right)$.
\item \emph{Case 2:  $\node$ is not the parent of $x$ and $y$ in $\tree$.} \\
  Let $\othernode$ be the child of $\node$ that lies on the path from $\node$ to $x$.
  Let $t$ be the distance between $x$ and $\bd P'(\othernode)$.
  Because of the shrinking performed in Step~\ref{step3}, we know that
  $t\geq (1/\log\Delta)\cdot (s/2)$, where $s$ is the length of the shortest edge
  of $P'(\othernode)$. Since $\RAR(P'(\othernode)) \leq 1/\eps$, we have
  $s\geq \eps^{1-1/d}\cdot\Vol(P'(\othernode))^{1/d}$. Hence,
\begin{align*}
  t &\geq \textstyle\frac{1}{\log\Delta} \cdot \frac{s}{2} \\
    &= \textstyle\Omega\left( \frac{\eps^{1-1/d}}{\log\Delta} \cdot \Vol(P'(\othernode))^{1/d} \right) \\
    &= \textstyle\Omega\left( \frac{\eps^{1-1/d}}{\log\Delta} \cdot (A^*_{\othernode})^{1/d} \right) \\
    &= \textstyle\Omega\left( \frac{1}{\log^2\Delta} \cdot (A^*_{\othernode})^{1/d} \right) \\
    &= \textstyle\Omega\left( \frac{1}{\log^2\Delta} \cdot l(\othernode) \cdot \frac{1}{\sqrt{d}} \right) \\
    &= \textstyle\Omega\left( \frac{D(x,y)}{\sqrt{d}\log^2\Delta}  \right).
\end{align*}%
\end{itemize}%
\end{proof}

Combining Lemmas~\ref{lemma:ultra_expansion}~and~\ref{lemma:ultra_contraction}
we obtain the main result of the section.

\begin{theorem}
For any fixed $d\geq 2$, there exists a polynomial-time, $O\left(\sqrt{d}\cdot \log^3\Delta\right)$-approx\-i\-ma\-tion algorithm 
for the problem of embedding ultrametrics into $\mathbb{R}^d$ with minimum distortion.
\end{theorem}

We remark that the running time of the above algorithm depends on the way the input is given.
If the ultrametric is given as a matrix of pairwise distances, then one needs to first compute a tree representation, with the points being the leaves of the tree.
For an $n$-point ultrametric, this task takes at least $\Omega(n^2)$ time.
Given such a tree representation, for any fixed dimension $d$, it is fairly easy to implement our algorithm in roughly quadratic time.
It is an interesting open problem to obtain an algorithm with near-linear running time.

As a final comment, note that while the approximation factor is relatively small (for instance, for $\Delta = n^{O(1)}$ and constant $d$,
it becomes $O(\log^3 n)$), the distortion itself can be much larger.
For example, embedding the uniform metric on $n$ vertices
(which is an ultrametric as well) into $\mathbb R^d$ requires distortion $\Omega(n^{1/d})$ for constant $d$.

\section*{Acknowledgments}
The authors wish to thank Roberto Tamassia for providing pointers to the applicability of 
polygonal partitions in visualization.

\bibliographystyle{abbrv}
\bibliography{bibfile}

\end{document}